\documentclass[letterpaper, 10 pt, conference]{ieeeconf}
\IEEEoverridecommandlockouts
\usepackage{cite}
\usepackage{amsmath,amssymb,amsfonts}
\usepackage{algorithmic}
\usepackage{graphicx}
\usepackage{textcomp}
\usepackage{xcolor}
\usepackage{mathrsfs}
\usepackage{psfrag}
\usepackage{cancel}

\makeatletter
\let\NAT@parse\undefined
\makeatother
\usepackage{hyperref}

\def\mc{\mathcal}

\begin{document}
\sloppy
\title{\LARGE \bf Making Nonlinear Systems Negative Imaginary via State Feedback}
\author{Kanghong Shi,$\quad$Ian R. Petersen, \IEEEmembership{Fellow, IEEE},$\quad$and$\quad$Igor G. Vladimirov 
\thanks{This work was supported by the Australian Research Council under grant DP190102158.}
\thanks{K. Shi, I. R. Petersen and I. G. Vladimirov are with the School of Engineering, College of Engineering and Computer Science, Australian National University, Canberra, Acton, ACT 2601, Australia.
        {\tt kanghong.shi@anu.edu.au}, {\tt ian.petersen@anu.edu.au}, {\tt igor.vladimirov@anu.edu.au}.}%
}

\newtheorem{definition}{Definition}
\newtheorem{theorem}{Theorem}
\newtheorem{conjecture}{Conjecture}
\newtheorem{lemma}{Lemma}
\newtheorem{remark}{Remark}
\newtheorem{corollary}{Corollary}
\newtheorem{assumption}{Assumption}

\maketitle

\thispagestyle{empty}
\pagestyle{empty}

\begin{abstract}
This paper provides a state feedback stabilization approach for nonlinear systems of relative degree less than or equal to two by rendering them nonlinear negative imaginary (NI) systems. Conditions are provided under which a nonlinear system can be made a nonlinear NI system or a nonlinear output strictly negative imaginary (OSNI) system. Roughly speaking, an affine nonlinear system that has a normal form with relative degree less than or equal to two, after possible output transformation, can be rendered nonlinear NI and nonlinear OSNI. In addition, if the internal dynamics of the normal form are input-to-state stable, then there exists a state feedback input that stabilizes the system. This stabilization result is then extended to achieve stability for systems with a nonlinear NI uncertainty.
\end{abstract}

\begin{keywords}
nonlinear negative imaginary systems, nonlinear output strictly negative imaginary systems, state feedback stabilization, robust control.
\end{keywords}

\section{INTRODUCTION}
Negative imaginary (NI) systems theory, which was introduced by Lanzon and Petersen in \cite{lanzon2008stability} and \cite{petersen2010feedback}, provides an approach to the robust control of flexible structures \cite{preumont2018vibration,halim2001spatial,pota2002resonant}. As the commonly used negative velocity feedback control \cite{brogliato2007dissipative} may not be suitable for some highly resonant systems, NI systems theory provides an alternative approach that uses positive feedback control. An NI system can be regarded as a positive real (PR) system cascaded with an integrator. Typical mechanical NI systems are systems with colocated force actuators and position sensors. Roughly speaking, a square real-rational proper transfer matrix $F(s)$ is said to be NI if is stable and $j(F(j\omega)-F(j\omega)^*)\geq 0$ for all $\omega \geq 0$. The fundamental stability results in NI systems theory are intuitive yet useful. Under mild assumptions, the positive feedback interconnection of an NI system $F(s)$ and a strictly negative imaginary (SNI) system $G(s)$ is internally stable if and only if the DC loop gain has all its eigenvalues strictly less than unity; i.e., $\lambda_{max}(F(0)G(0))<1$. Since it was introduced in 2008 \cite{lanzon2008stability}, NI systems theory has attracted attention from many control theorists \cite{xiong2010negative,song2012negative,mabrok2014generalizing,wang2015robust,bhowmick2017lti} and has been applied in many fields including nano-positioning control \cite{mabrok2013spectral,das2014mimo,das2014resonant,das2015multivariable} and the control of lightly damped structures \cite{cai2010stability,rahman2015design,bhikkaji2011negative}, etc.

NI systems theory was extended to nonlinear systems in \cite{ghallab2018extending,shi2021robust,ghallab2022negative} through the notion of counterclockwise input-output dynamics \cite{angeli2006systems}. Roughly speaking, a system is said to be nonlinear NI if it has a positive definite storage function and is dissipative with respect to the supply rate $u^T\dot y$, where $u$ and $y$ are the input and output of the system, respectively. This definition is generalized in \cite{shi2020robustc}, to only require positive semidefiniteness of the storage function in order to allow for systems with poles at the origin; e.g., single and double integrators. Also introduced in \cite{shi2021robust} and \cite{shi2020robustc} is the notion of nonlinear output strictly negative imaginary (OSNI) systems (see \cite{bhowmick2017lti} and \cite{bhowmickoutput} for the definition of linear OSNI systems). Under the control of suitable nonlinear OSNI controllers, nonlinear NI systems can be asymptotically stabilized under mild assumptions. An advantage of linear and nonlinear NI systems theory is that NI systems can have relative degree of zero, one and two, in comparison to PR and passive systems whose relative degree can only be zero or one. With this advantage, NI systems theory can provide stability for systems of relative degree less than or equal to two, which cannot be dealt with by passivity or PR systems theory. One such example arises in the problem of state feedback stabilization.

Passivity and PR systems theory is applied in many papers to achieve state feedback stabilization \cite{kokotovic1989positive,saberi1990global,byrnes1991passivity,byrnes1991asymptotic,santosuosso1997passivity,lin1995feedback,jiang1996passification}. The general idea applied in these papers is to render part of a nonlinear system PR or passive using state feedback. Then stability can be obtained using the passivity or PR properties of the resulting system. Such state feedback passivity results are significant not only because they provide a generalization to the feedback linearization method, but also because feedback analysis design for passive systems is comparatively simple and intuitive \cite{byrnes1991passivity}. However, due to the nature of passive systems, a common assumption made in these papers is that the systems in question must have relative degree one. This restriction rules out a wide variety of systems which have output entries of relative degree two.

Since NI systems theory can deal with systems with relative degree zero, one and two, it is useful as a complement to passivity and PR systems theory in addressing the state feedback stabilization problem. In \cite{shi2021negative}, conditions are given for linear time-invariant (LTI) systems with relative degree one and relative degree two to be rendered NI or strongly strictly NI (SSNI) using state feedback control. This result is then generalized in the paper \cite{shi2021necessary}, which gives necessary and sufficient conditions under which an LTI system is state feedback equivalent to an NI, OSNI or SSNI system. In \cite{shi2021necessary}, the system is allowed to have mixed relative degree one and two. Stabilization results are also provided in \cite{shi2021negative} and \cite{shi2021necessary} for systems with SNI uncertainties.

Considering the nonlinear nature of most control systems, this paper investigates the problem of making affine nonlinear systems nonlinear NI using state feedback, in order to provide a method of stabilization for nonlinear systems of relative degree less than or equal to two. This paper provides conditions under which a system can be rendered nonlinear NI or OSNI, as well as corresponding formulas for the control inputs. Roughly speaking, if an affine nonlinear system of relative degree less than or equal to two can be transformed into a normal form (see \cite{isidori1995nonlinear,byrnes1991asymptotic} for a description of the normal form), then there exists state feedback control such that the resulting system is NI or OSNI. If in addition, the internal dynamics in the normal form are input-to-state stable (ISS), then there exists a state feedback control that stabilizes the system. Furthermore, such a system with a nonlinear NI plant uncertainty can also be stabilized if in addition there exists a storage function for this system, which is positive definite with respect to a specific subset of the state variables.

From the technical point of view, the contribution of this work is providing an alternative approach to the previous passivity-based state feedback stabilization results (e.g., \cite{byrnes1991passivity}) to overcome their limitations by allowing systems with output entries of relative degree two. More importantly, it broadens the class of systems to which nonlinear NI systems theory is applicable.

This paper is organized as follows: Section \ref{sec:pre} reviews the essential nonlinear NI systems definitions. Section \ref{sec:initial stability} provides conditions to render a system nonlinear NI or OSNI, locally and globally, as shown in Theorems \ref{theorem:render NI} and \ref{theorem:render globally NI}. Using the nonlinear NI properties, Theorems \ref{theorem:locally stabilization with 0 new input} and \ref{theorem:globally stabilization with 0 new input} provide conditions and formulas for state feedback stabilization, locally and globally. In Section \ref{sec:NI uncertainty}, Theorems \ref{theorem:stabilization} and \ref{theorem:global stabilization} provide conditions and formulas for the state feedback stabilization of a system with a nonlinear NI uncertainty. The process of stabilizing such an uncertain system is illustrated in a numerical example in Section \ref{sec:example}. A conclusion is given in Section \ref{sec:conclusion}.

Notation: The notation in this paper is standard. $\mathbb R$ denotes the fields of real numbers. $\mathbb R^{m\times n}$ denotes the space of real matrices of dimension $m\times n$. $A^T$ denotes the transpose of a matrix $A$. $A^{-T}$ denotes the transpose of the inverse of $A$; i.e., $A^{-T}=(A^{-1})^T=(A^T)^{-1}$. $\lambda_{max}(A)$ denotes the largest eigenvalue of a matrix $A$ with real spectrum. $\|\cdot\|$ denotes the standard Euclidean norm. $C^k$ represents the class of $k$-time continuously differentiable functions. Given a scalar function $h(x)$ and a vector field $f(x)$, $L_fh(x)$ denotes the Lie derivative of $h(x)$ with respect to $f(x)$; i.e., $L_fh(x):=\frac{\partial h(x)}{\partial x}f(x)$. For two vector fields $f$ and $g$ on $D\subset \mathbb R^n$, the Lie bracket $[f,g]$ is a third vector field defined by $$[f,g](x)=\frac{\partial g}{\partial x}f(x)-\frac{\partial f}{\partial x}g(x),$$ where $\frac{\partial g}{\partial x}$ and $\frac{\partial f}{\partial x}$ are Jacobian matrices. Repeated bracketing of $g$ with $f$ can be represented using the following adjoint representation for simplicity:
\begin{align}
ad_f^0g(x) = &\ g(x),\notag\\
ad_fg(x) = &\ [f,g](x),\notag\\
ad_f^kg(x) = &\ [f,ad_f^{k-1}g](x), \quad k\geq 1.\notag	
\end{align}

\section{PRELIMINARIES}\label{sec:pre}

Consider the following general affine nonlinear system:
\begin{subequations}\label{eq:general system}
\begin{align}
\Sigma: \quad   \dot x=&\ f(x)+g(x)u,\label{eq:general system state}\\
    y=&\ h(x),\label{eq:general system output}
\end{align}	
\end{subequations}
where $x\in \mathbb R^n$, $u\in \mathbb R^p$ and $y\in \mathbb R^p$ are the state, input and output of the system. The function $f:\mathbb R^n \to \mathbb R^n$ is globally Lipschitz, $g:\mathbb R^n \to \mathbb R^{n\times p}$ and $h:\mathbb R^n \to \mathbb R^p$. Here, $f,h$ and the columns $g^1,\cdots,g^p$ are of class $C^\infty$. We suppose that the vector field $f$ has at least one equilibrium. Then without loss of generality, we can assume $f(0)=0$ and $h(0)=0$ after a coordinate shift.
\begin{definition}[Nonlinear NI Systems]\label{def:nonlinear NI}\cite{shi2020robustc,ghallab2018extending}
The system (\ref{eq:general system}) is said to be a nonlinear negative imaginary (NI) system if there exists a positive semidefinite storage function $V:\mathbb R^n\to \mathbb R$ of class $C^1$ such that
\begin{equation}\label{eq:NI MIMO definition inequality}
    \dot V(x(t))\leq u(t)^T\dot y(t)
\end{equation}
for all $t\geq 0$.
\end{definition}

\begin{definition}[Nonlinear OSNI Systems]\label{def:nonlinear OSNI}
The system (\ref{eq:general system}) is said to be a nonlinear output strictly negative imaginary (OSNI) system if there exists a positive semidefinite storage function $V:\mathbb R^n\to\mathbb R$ of class $C^1$ and a scalar $\epsilon>0$ such that
\begin{equation}\label{eq:dissipativity of OSNI}
    \dot V(x(t))\leq u(t)^T\dot { y}(t) -\epsilon \left\|\dot {y}(t)\right\|^2
\end{equation}
for all $t\geq 0$. In this case, we also say that system (\ref{eq:general system}) is nonlinear OSNI with degree of output strictness $\epsilon$.
\end{definition}

\begin{definition}[Vector Relative Degree]\label{def:vector relative degree}\cite{isidori1995nonlinear}
A multivariable nonlinear system of the form (\ref{eq:general system}) has vector relative degree $\{r_1,\cdots,r_m\}$ at a point $x^\circ$ if\\
(i)
\begin{equation*}
L_gL_f^kh_i(x)=0	
\end{equation*}
for all $k<r_i-1$, for all $1\leq i\leq p$ and for all $x$ in a neighbourhood of $x^\circ$,\\
(ii) the $p\times p$ matrix
\begin{equation}\label{eq:A(x)}
A(x)=\left[\begin{matrix}L_gL_f^{r_1-1}h_1(x)\\ \vdots \\ L_gL_f^{r_p-1}h_p(x) \end{matrix}\right]
\end{equation}
is nonsingular at $x=x^\circ$. Here $h_i(x)$ denotes the $i$-th entry of the vector $h(x)\in \mathbb R^p$.
\end{definition}

\begin{definition}[Uniform Relative Degree]
The system (\ref{eq:general system}) is said to have uniform relative degree $\{r_1,\cdots,r_p\}$ if it has vector relative degree $\{r_1,\cdots,r_p\}$ at all $x\in R^n$.
\end{definition}

\begin{definition}[Class $\mc{K}$ and $\mc{K}_{\infty}$ Functions]\cite{khalil2002nonlinear}
A continuous function $\alpha:[0,a)\to [0,\infty)$ is said to belong to class $\mc{K}$ if it is strictly increasing and $\alpha(0)=0$. It is said to belong	to class $\mc{K}_{\infty}$ if $a=\infty$ and $\alpha(r)\to \infty$ as $r\to \infty$.
\end{definition}

\begin{definition}[Class $\mc{KL}$ Functions]\cite{khalil2002nonlinear}
A continuous function $\beta:[0,a)\times [0,\infty)\to [0,\infty)$ is said to belong to class $\mc{KL}$ if for each fixed $s$, the mapping $\beta(r,s)$ belongs to class $\mc K$ with respect to $r$ and for each fixed $r$, the mapping $\beta(r,s)$ is decreasing with respect to $s$ and $\beta(r,s)\to 0$ as $s\to \infty$.
\end{definition}

\section{STABILIZATION OF A SYSTEM USING NONLINEAR NI SYSTEMS THEORY}\label{sec:initial stability}
Let us consider a nonlinear system of the form (\ref{eq:general system}). We aim to stabilize this system by rendering it a nonlinear OSNI system as per Definition \ref{def:nonlinear OSNI} in the case that the system has relative degree less than or equal to two. We now provide a formal definition of systems of relative degree less than or equal to two.
\begin{definition}
	A system of the form (\ref{eq:general system}) is said to have relative degree less than or equal to two if it has a vector relative degree $r=\{r_1,\cdots,r_p\}$, where $1\leq r_i \leq 2$ for all $i=1,\cdots,p$. Without loss of generality, assume the components of the output vector are sorted such that the components in the vector relative degree are in nondecreasing order; i.e., $r_i = 1$ for $i=1,2,\cdots,p_1$ and $r_i = 2$ for $i=p_1+1,p_1+2,\cdots,p$, where $p_1$ is the number of ones in the vector relative degree $r$.
\end{definition}

The paper \cite{byrnes1991asymptotic} provides conditions for a system with a vector relative degree to have local and global normal forms. Here, we focus on the specific case that the system has relative degree less than or equal to two.
\begin{lemma}\label{lemma:normal form conditions}(see also \cite{byrnes1991asymptotic})
Suppose the system (\ref{eq:general system}) has relative degree less than or equal to two at $x=0$. If the distribution
\begin{equation*}\label{eq:G}
G = span\{g^1,g^2,\cdots,g^p\}	
\end{equation*}
is involutive, then the system (\ref{eq:general system}) can be described locally around $x=0$ by the following normal form
\begin{subequations}\label{eq:normal form}
\begin{align}
\Sigma:\quad\dot z =&\ f^*(z,\xi),\label{eq:normal form a}\\
\dot \xi_1 = &\ a_1(z,\xi)+b_1(z,\xi)	u,\\
\dot \xi_2 = &\ \xi_3,\\
\dot \xi_3 =&\ a_2(z,\xi)+b_2(z,\xi)u,\\
y =& \left[\begin{matrix}\xi_1\\ \xi_2\end{matrix}
\right]
\end{align}	
\end{subequations}
where $u$ and $y$ are still the input and output of the system. The vector $[z^T\ \xi^T]^T$ is the new state of the system, where $z\in \mathbb R^m$ ($m\geq 0$) and $\xi = [\xi_1^T\ \xi_2^T\ \xi_3^T]^T$. The vector $\xi_1 \in \mathbb R^{p_1}$ contains the state vector entries corresponding to the ones in the vector relative degree. The vector $\xi_2\in \mathbb R^{p_2}$ ($p_2:=p-p_1$) contains the state vector entries corresponding to the twos in the vector relative degree and $\xi_3\in \mathbb R^{p_2}$ is defined as the derivative of $\xi_2$. Here, $f^*,a_1,b_1,a_2,b_2$ are functions of suitable dimensions. Also,
\begin{equation*}
a_1(z,\xi)=	\left[\begin{matrix}
L_fh_1(x)\\ \vdots \\ L_fh_{p_1}(x)	
\end{matrix}
\right], \quad a_2(z,\xi)=	\left[\begin{matrix}
L_f^2h_{p_1+1}(x)\\ \vdots \\ L_f^2h_p(x)	
\end{matrix}
\right],
\end{equation*}
and
\begin{equation*}
	b_1(z,\xi) = \left[\begin{matrix}
L_gh_1(x)\\ \vdots \\ L_gh_{p_1}(x)	
\end{matrix}
\right], \quad b_2(z,\xi) = \left[\begin{matrix}
L_gL_fh_{p_1+1}(x)\\ \vdots \\ L_gL_fh_p(x)
\end{matrix}
\right].
\end{equation*}
Hence,
\begin{equation*}
\left[\begin{matrix}b_1(z,\xi)\\ b_2(z,\xi)\end{matrix}
\right]=A(x)	
\end{equation*}
as in (\ref{eq:A(x)}) and is nonsingular for $(z,\xi)$ at $(0,0)$.
\end{lemma}
\begin{proof}
The proof directly follows from Propositions 3.2a and 3.2b in \cite{byrnes1991asymptotic} as this lemma is a special case in which the vector relative degree only contains the numbers one and two. Note that there are differences between the notation used here and used in \cite{byrnes1991asymptotic}.
\end{proof}

In the normal form (\ref{eq:normal form}), the dynamics described by (\ref{eq:normal form a}) are called the \emph{internal dynamics} of the system. When the output is identically zero, the internal dynamics are called the \emph{zero dynamics}\cite{byrnes1991asymptotic,khalil2002nonlinear,isidori1995nonlinear}. In the case of system (\ref{eq:normal form}), $y$ being identically zero implies $\xi=0$. Therefore, the zero dynamics are described by
\begin{equation*}\label{eq:zero dynamics}
\dot z = f^*(z,0).	
\end{equation*}

\begin{lemma}(see \cite{byrnes1991asymptotic})\label{lemma:global normal form}
The system (\ref{eq:general system}) is globally diffeomorphic to a system having the normal form (\ref{eq:normal form}) if:\\
\textbf{H1}: the system has uniform relative degree less than or equal to two;\\
\textbf{H2}: the vector fields
\begin{equation*}
X_i^k = ad_{\tilde f}^{k-1}\tilde g_i, \quad 1\leq i\leq p, \quad 1\leq k\leq r_i
\end{equation*}
are complete;\\
\textbf{H3}: $[X_i^1, X_j^1]=0$ for all $1\leq i,j\leq p$.\\
Here,
\begin{equation*}
\tilde f = f-gA^{-1}(x)\left[\begin{matrix}L_f^{r_1}h_1(x)\\ \vdots \\ L_f^{r_p}h_p(x)
\end{matrix}
\right], \quad \tilde g=gA^{-1}(x).
\end{equation*}
\end{lemma}
\begin{proof}
See Corollary 5.6 in \cite{byrnes1991asymptotic}.	
\end{proof}

\begin{lemma}\label{lemma:rendering NI(OSNI)}
Consider the system (\ref{eq:normal form}) where $\left[\begin{matrix}b_1(z,\xi)\\ b_2(z,\xi)\end{matrix}
\right]$ is nonsingular. Then it can be rendered a nonlinear NI system as in Definition \ref{def:nonlinear NI} using the state feedback control law
\begin{equation}\label{eq:NI state feedback control}
u = \left[\begin{matrix}b_1(z,\xi)\\ b_2(z,\xi)\end{matrix}
\right]^{-1}\left(v-\left[\begin{matrix}a_1(z,\xi)+\left(\frac{\partial V_1(\xi_1)}{\partial \xi_1}\right)^T \\ a_2(z,\xi)+\left(\frac{\partial V_2(\xi_2)}{\partial \xi_2}\right)^T+\lambda \xi_3\end{matrix}
\right]\right),	
\end{equation}
where $v\in \mathbb R^p$ is the new input, $V_1(\xi_1)$ and $V_2(\xi_2)$ can be any positive semi-definite functions, and $\lambda\geq 0$ is a scalar. Moreover, if $\lambda>0$, then the resulting system is a nonlinear OSNI system as per Definition \ref{def:nonlinear OSNI} with degree of output strictness $\epsilon=\min\{1,\lambda\}$. The storage function of the nonlinear NI (OSNI) system is
\begin{equation}\label{eq:rendered NI storage function}
V(z,\xi)=\tilde V(\xi)= V_1(\xi_1)+V_2(\xi_2)+\frac{1}{2}\xi_3^T\xi_3.	
\end{equation}
\end{lemma}
\begin{proof}
Let $v=[v_1^T \ v_2^T]^T$ in (\ref{eq:NI state feedback control}), where $v_1\in \mathbb R^{p_1}$ and $v_2 \in \mathbb R^{p_2}$. With the state feedback control (\ref{eq:NI state feedback control}), the system (\ref{eq:normal form}) now becomes
\begin{subequations}\label{eq:system rendered NI}
\begin{align}
\dot z =&\ f^*(z,\xi),\label{eq:system rendered NI a}\\
\dot \xi_1 = &\ v_1-\left(\frac{\partial V_1(\xi_1)}{\partial \xi_1}\right)^T,\label{eq:system rendered NI b}\\
\dot \xi_2 = &\ \xi_3,\\
\dot \xi_3 =&\ v_2-\left(\frac{\partial V_2(\xi_2)}{\partial \xi_2}\right)^T-\lambda \xi_3,\label{eq:system rendered NI d}\\
y =& \left[\begin{matrix}\xi_1\\ \xi_2\end{matrix}
\right].
\end{align}	
\end{subequations}
The nonlinear NI inequality is satisfied for the resulting system with  the positive semidefinite storage function (\ref{eq:rendered NI storage function}), which is shown in the following:
\begin{align}
\dot V&(z,\xi)-v^T\dot y\notag\\
 =&	\frac{\partial V(z,\xi)}{\partial z} \dot z + \frac{\partial V(z,\xi)}{\partial \xi_1} \dot \xi_1 + \frac{\partial V(z,\xi)}{\partial \xi_2} \dot \xi_2 + \frac{\partial V(z,\xi)}{\partial \xi_3} \dot \xi_3\notag\\
& - v^T\dot y\notag\\
=&\ 0 + \frac{\partial V_1(\xi_1)}{\partial \xi_1}\left(v_1-\left(\frac{\partial V_1(\xi_1)}{\partial \xi_1}\right)^T\right)+ \frac{\partial V_2(\xi_2)}{\partial \xi_2} \xi_3\notag\\
&  + \xi_3^T \left(v_2-\left(\frac{\partial V_2(\xi_2)}{\partial \xi_2}\right)^T-\lambda \xi_3\right)\notag\\
&-v_1^T\left(v_1-\left(\frac{\partial V_1(\xi_1)}{\partial \xi_1}\right)^T\right)-v_2^T \xi_3\notag\\
=& -\left(v_1^T-\frac{\partial V_1(\xi_1)}{\partial \xi_1}\right)\left(v_1-\left(\frac{\partial V_1(\xi_1)}{\partial \xi_1}\right)^T\right)-\lambda \xi_3^T\xi_3\notag\\
=& -\|\dot \xi_1\|^2-\lambda\|\dot \xi_2\|^2\notag\\
\leq & -\epsilon \| \dot y \|^2,\label{eq:nonlinear OSNI ineq of the rendered system}
\end{align}
where $\epsilon=\min\{1,\lambda\}$. The system is a nonlinear NI system as per Definition \ref{def:nonlinear NI}. Moreover, if $\lambda>0$, then $\epsilon>0$. In this case, the system is a nonlinear OSNI system and $\epsilon$ is the degree of output strictness of the system.	
\end{proof}
\begin{theorem}\label{theorem:render NI}
Suppose the system (\ref{eq:general system}) has relative degree less than or equal to two at $x=0$ and the distribution
\begin{equation*}\label{eq:G}
G = span\{g^1,g^2,\cdots,g^p\}	
\end{equation*}
is involutive. Then the system (\ref{eq:general system}) can be rendered a nonlinear NI (OSNI) system locally around $x=0$ using the state feedback control
\begin{align}\label{eq:NI state feedback control in theorem}
u = &A(x)^{-1}\left(v-\left[\begin{matrix}L_f^{r_1}h_1(x)\\ \vdots \\ L_f^{r_p}h_p(x) \end{matrix}\right]
-\left[\begin{matrix}\left(\frac{\partial V_1(\xi_1)}{\partial \xi_1}\right)^T \\ \left(\frac{\partial V_2(\xi_2)}{\partial \xi_2}\right)^T+\lambda \xi_3\end{matrix}
\right]\right),	
\end{align}
where $A(x)$ is defined in (\ref{eq:A(x)}), $v\in \mathbb R^p$ is the new input, $V_1(\xi_1)$ and $V_2(\xi_2)$ can be any positive semi-definite functions, and $\lambda\geq 0$ ($\lambda > 0$) is a scalar. Also, the function (\ref{eq:rendered NI storage function}) is a storage function for the resulting nonlinear NI (OSNI) system.
\end{theorem}
\begin{proof}
	This result directly follows from Lemmas \ref{lemma:normal form conditions} and \ref{lemma:rendering NI(OSNI)}.
\end{proof}

\begin{theorem}\label{theorem:render globally NI}
Suppose the system (\ref{eq:general system}) satisfies H1, H2 and H3. Then the system (\ref{eq:general system}) can be globally rendered a nonlinear NI (OSNI) system using the state feedback control
(\ref{eq:NI state feedback control in theorem}). Also, the function $V(z,\xi)$ defined in (\ref{eq:rendered NI storage function}) is a storage function for the resulting nonlinear NI (OSNI) system.
\end{theorem}
\begin{proof}
See the proof of Theorem \ref{theorem:render NI}, using Lemma \ref{lemma:global normal form} instead of Lemma \ref{lemma:normal form conditions}.	
\end{proof}

Considering the restriction in condition (ii) of Definition \ref{def:vector relative degree}, some systems do not have vector relative degrees. However, for a system that does not have a vector relative degree, sometimes there exists an output transformation that transforms it into a system with a vector relative degree.
We generalize Theorem \ref{theorem:render NI} by showing that the result is invariant to a nonsingular output transformation.
\begin{lemma}\label{lemma:output transformation invariant NI}
A system with output $y$ and input $u$ is a nonlinear NI (OSNI) system if and only if the system with output $\tilde y=T_y y$ and input $\tilde u=T_y^{-T}u$ is a nonlinear NI (OSNI) system. Here $T_y$ is a nonsingular constant matrix.
\end{lemma}
\begin{proof}
Considering that
\begin{equation*}
\tilde u^T \dot {\tilde y} = u^T T_y^{-1}T_y \dot y = u^T\dot y,	
\end{equation*}
the nonlinear NI inequality (\ref{eq:NI MIMO definition inequality}) is satisfied for one of these two systems if and only if it is satisfied for the other. Also, we have that
\begin{equation*}
\|\dot {\tilde y}\|^2=\dot y^T T_y^TT_y \dot y \leq \lambda_{max}(T_y^TT_y) \|\dot y\|^2.
\end{equation*}
Therefore, if the system with input $u$ and output $y$ is a nonlinear OSNI system, then there exists a positive semidefinite storage function $V(x)$ such that
\begin{align}
\dot V(x) \leq & u^T \dot y -\epsilon \|\dot y\|^2 \notag\\
\leq & \tilde u^T \dot {\tilde y}-\frac{\epsilon}{\lambda_{max}(T_y^TT_y)}\|\dot {\tilde y}\|^2.\notag
\end{align}
This implies that the system with input $\tilde u$ and output $\tilde y$ is also a nonlinear OSNI system. The sufficiency part can be proved similarly by considering the inverses of the transformations.
\end{proof}
\begin{lemma}\label{lemma:invariant equivalent NI}
Consider a system of the form (\ref{eq:general system}) and an output transformation $\tilde y = T_y y$ where $T_y\in \mathbb R^{p\times p}$ is a nonsingular constant matrix. If there exists a state feedback control law
\begin{equation*}
u = k_x(x)x+k_u(x)v,
\end{equation*}
under which the system with the new input $v\in \mathbb R^p$ is a nonlinear NI (OSNI) system, then the output transformed system; i.e., the system with input $u$ and output $\tilde y = T_y y$, can also be rendered a nonlinear NI (OSNI) system using the state feedback control law
\begin{equation*}
u = k_x(x)x + k_u(x)T_y^T \tilde v	,
\end{equation*}
where $\tilde v \in \mathbb R^p$ is the new input.
\end{lemma}
\begin{proof}
According to Lemma \ref{lemma:invariant equivalent NI}, the system with input $v$ and output $y$ is nonlinear NI (OSNI) if and only if the system with input $\tilde v = T_y^{-T}v$ and output $\tilde y = T_y y$ is nonlinear NI (OSNI). This completes the proof.
\end{proof}
\begin{corollary}\label{corollary:render output transformed systems NI}
Suppose the system (\ref{eq:general system}) can be output transformed into a system with relative degree less than or equal to two at $x=0$ using the output transformation 
\begin{equation}\label{eq:output transformation}
\tilde y = T_y y,
\end{equation}
where $T_y\in \mathbb R^{p\times p}$ is a nonsingular constant matrix. Also, suppose the distribution
\begin{equation*}\label{eq:G}
G = span\{g^1,g^2,\cdots,g^p\}	
\end{equation*}
is involutive. Then the system (\ref{eq:general system}) can be rendered a nonlinear NI (OSNI) system locally around $x=0$ using state feedback control.
\end{corollary}
\begin{proof}
The proof follows directly from Theorem \ref{theorem:render NI} and Lemma \ref{lemma:invariant equivalent NI}.	
\end{proof}

\begin{corollary}\label{corollary:render globally output transformed systems NI}
Suppose the system (\ref{eq:general system}) can be output transformed into a system satisfying H1, H2 and H3 using the output transformation (\ref{eq:output transformation}). Then the system (\ref{eq:general system}) can be globally rendered a nonlinear NI (OSNI) system using state feedback control.
\end{corollary}
\begin{proof}
The proof follows directly from Theorem \ref{theorem:render globally NI} and Lemma \ref{lemma:invariant equivalent NI}.	
\end{proof}

\begin{remark}
In the state feedback control laws applied in Theorem \ref{theorem:render NI} and Corollary \ref{corollary:render output transformed systems NI}, the only choice of $\lambda$ that makes the resulting system a nonlinear NI system but not a nonlinear OSNI system is $\lambda=0$. Since full state feedback is available, we can simply choose $\lambda>0$ and render the system a nonlinear OSNI system in order to obtain more strictness. This strictness is helpful in the stabilization of the system.
\end{remark}

If the system (\ref{eq:general system}) is rendered a nonlinear OSNI system with a storage function which is positive definite with respect to $\xi$, then according to the dissipation inequality (\ref{eq:dissipativity of OSNI}), giving zero input to the system will result the boundedness of the state $\xi$. Indeed, as will be shown later, the state $\xi$ will converge to zero under zero input. Given the stability of the state $\xi$, we consider the question of what additional conditions are needed in order to make the state $z$ also stable. First, we need to provide several definitions regarding to the internal dynamics in the normal form of a nonlinear system.

\begin{definition}[(Globally) Minimum Phase]\label{def:minimum phase}\cite{isidori1995nonlinear,byrnes1991passivity}
A nonlinear system which has a normal form is said to be (globally) minimum phase if its zero dynamics have a (globally) asymptotically stable equilibrium at the origin.
\end{definition}

In the case of system (\ref{eq:normal form}), the minimum phase property guarantees that when $\xi=0$, if $z$ is finite, it will also converge to zero. In other words, if we view the state $\xi$ as the input to the internal dynamics $\dot z = f^*(z,\xi)$, then (global) minimum phase property implies that the internal dynamics is (globally) asymptotically stable with zero input, namely 0-AS (0-GAS) for short. However, examples in \cite{sontag1989smooth,sontag2008input,khalil2002nonlinear} show that for systems that are 0-AS (0-GAS), its state may diverge under a bounded input that converges to zero. This phenomena motivated the concept of input-to-state stability (ISS) \cite{sontag1989smooth,sontag2008input}. As is discussed in \cite{liberzon2002output}, the asymptotic stability of the zero dynamics is sometimes insufficient for control design purposes until it is combined with the ISS property of the internal dynamics. This is a common requirement (see for example \cite{praly1993stabilization}). Let us now recall the definitions of ISS and locally ISS (LISS) systems.

To avoid introducing new system models, let us consider the system of the form (\ref{eq:normal form a}). Let us rewrite it in the following as a seperate system:
\begin{equation}\label{eq:system model ISS definition}
\dot z = f^*(z,\xi),	
\end{equation}
where $\xi$ acts as the input to this system.

\begin{definition}[Input-to-State Stability]\label{def:ISS}\cite{sontag1989smooth,sontag2008input,khalil2002nonlinear,liberzon2003switching,isidori1995nonlinear}
The system (\ref{eq:system model ISS definition}) is said to be input-to-state stable (ISS) if there exist a class $\mc{KL}$ function $\beta$ and a class $\mc K$ function $\gamma$ such that for any initial state $z(t_0)$ and any bounded input $\xi(t)$, the solution $z(t)$ exists for all $t\geq t_0$ and satisfies
\begin{equation}\label{eq:ISS ineq}
\|z(t)\| \leq \beta(\|z(t_0)\|,t-t_0)+\gamma\left(\sup_{t_0\leq\tau\leq t}\|\xi(\tau)\|\right).	
\end{equation}
\end{definition}

\begin{definition}[Locally Input-to-State Stability]\label{def:LISS}\cite{sontag1996new}
The system (\ref{eq:system model ISS definition}) is said to be locally input-to-state stable (LISS) if there exist a class $\mc{KL}$ function $\beta$, a class $\mc K$ function $\gamma$ and constants $\rho_z,\rho_{\xi}>0$ such that for any initial state $z(t_0)$ with $\|z(t_0)\|\leq \rho_z$ and any bounded input $\xi(t)$ with $\sup_{t_0\leq\tau\leq t}\|\xi(\tau)\|\leq \rho_{\xi}$ for all $t\geq t_0$, the solution $z(t)$ exists and satisfies (\ref{eq:ISS ineq}) for all $t\geq t_0$.
\end{definition}

\begin{lemma}\label{lemma:ISS x convergence}\cite{khalil2002nonlinear,liberzon2003switching}
Suppose the system (\ref{eq:system model ISS definition}) is ISS. If $\xi(t)\to 0$ as $t\to \infty$, so does $z(t)$.	
\end{lemma}
\begin{proof}(See also Exercise 4.58 in \cite{khalil2002nonlinear} and its solution manual).
We show that for any $\epsilon>0$, there exists a $T>0$ such that $\|z(t)\| \leq \epsilon$, $\forall t\geq T$. Since $\gamma$ is a class $\mc K$ function, then given $\epsilon>0$, there exists $\epsilon_1>0$ such that $\gamma (\epsilon_1)\leq \frac{\epsilon}{2}$. Since $\lim_{t\to \infty}\xi(t) = 0$,	given $\epsilon_1$ there is a $T_1>0$ such that $\|\xi(t)\|\leq \epsilon_1$ for $t\geq T_1$. Take $t_1>T_1$. Then for $t>t_0$, we have
\begin{align}
	\|z(t)\| \leq & \beta(\|z(t_1)\|,t-t_1)+\gamma\left(\epsilon_1\right)\notag\\
	\leq & \beta(c,t-t_1)+\frac{\epsilon}{2},\notag
\end{align}
where $c=\|z(t_1)\|$ is a constant. Since $\beta$ is a class $\mc {KL}$ function and $\|z(t_1)\|$ is bounded, then $\|\beta(c,t-t_1)\|\to 0$ as $t\to 0$. There exists a $T_2>0$ such that $\|\beta(c,t-t_1)\|\leq \frac{\epsilon}{2}$, for all $t>T_2$. Thus,
\begin{equation*}
\|z(t)\|\leq \epsilon, \forall t\geq T = \max\{T_1,T_2\},	
\end{equation*}
which shows that $\lim_{t\to\infty}z(t)\to 0$.
\end{proof}

\begin{lemma}\label{lemma:LISS x convergence}\cite{khalil2002nonlinear,liberzon2003switching}
If the system (\ref{eq:system model ISS definition}) is LISS, then there exist constants $\tilde\rho_z,\tilde\rho_{\xi}>0$ such that for any initial state $z(t_0)$ with $\|z(t_0)\|\leq \tilde\rho_z$ and any bounded input $\xi(t)$ with $\sup_{t_0\leq\tau < \infty}\|\xi(\tau)\|\leq \tilde\rho_{\xi}$ and $\xi(t)\to 0$ as $t \to \infty$, we have $z(t)\to 0$ as $t\to\infty$.	
\end{lemma}
\begin{proof}
Suppose the system is LISS, then there exist $\rho_z$ and $\rho_{\xi}$ such that for any initial state $z(t_0)$ with $\|z(t_0)\|\leq \rho_z$ and any bounded input $u(t)$ with $\sup_{t_0\leq\tau\leq t}\|\xi(\tau)\|\leq \rho_{\xi}$ for all $t\geq t_0$, we have that $z(t)$ exists and satisfies (\ref{eq:ISS ineq}). Considering that $\gamma$ is a class $\mc K$ function, then there exists $\tilde\rho_{\xi}\leq \rho_{\xi}$ such that
\begin{equation*}
\gamma\left(\tilde\rho_{\xi}\right) \leq \frac{\rho_z}{2}.
\end{equation*}
Also, considering $\beta$ is a class $\mc{KL}$ function, there exists $\tilde\rho_z\leq \rho_z$ such that
\begin{equation*}
\beta\left(\tilde\rho_z,0\right)	 \leq \frac{\rho_z}{2}.
\end{equation*}
Let $\|z(t_0)\|\leq \tilde\rho_z$ and $\sup_{t_0\leq\tau < \infty}\|\xi(\tau)\|\leq \tilde\rho_{\xi}$. According to the LISS property, for all $t\geq t_0$, $z(t)$ exists and satisfies
\begin{align}
	\|z(t)\| \leq & \beta(\|z(t_0)\|,t-t_0)+\gamma\left(\sup_{t_0\leq\tau\leq t}\|\xi(\tau)\|\right)\notag\\
	\leq & \beta\left(\tilde\rho_z,0\right)+\gamma\left(\tilde\rho_{\xi}\right)\notag\\
	\leq & \rho_z.\notag\label{eq:ineq in LISS lemma proof}
\end{align}
Since $\|z(t)\|\leq \rho_z$, the property in (\ref{eq:ISS ineq}) still holds if $t_0$ is substituted by any $t_1\geq t_0$. The rest of the proof follows as in the proof of Lemma \ref{lemma:ISS x convergence}, using Definition \ref{def:LISS} instead of Definition \ref{def:ISS}.
\end{proof}

As is proved in \cite{sontag1996new}, ISS implies 0-GAS and LISS implies 0-AS. We now introduce a new version of the (global) minimum phase property in the following, in which the 0-AS (0-GAS) requirement is replaced by an LISS (ISS) requirement. Hence, the new definition is stricter than Definition \ref{def:minimum phase}.

\begin{definition}[(Globally) Strictly Minimum Phase]
A system is said to be (Globally) strictly minimum phase if its internal dynamics of the form
\begin{equation*}
\dot z = f^*(z,\xi)
\end{equation*}
are LISS (ISS) with respect to $\xi$.
\end{definition}

\begin{theorem}\label{theorem:locally stabilization with 0 new input}
After possible output transformation (\ref{eq:output transformation}), suppose the system (\ref{eq:general system}) satisfies the following:

(i). it has relative degree less than or equal to two at $x=0$;

(ii). the distribution
\begin{equation*}\label{eq:G}
G = span\{g^1,g^2,\cdots,g^p\}	
\end{equation*}
is involutive;

(iii). the system (\ref{eq:general system}) is strictly minimum phase around $x=0$. Then the system (\ref{eq:general system}) can be locally asymptotically stabilized using the state feedback control law
\begin{align}\label{eq:stabilizing state feedback}
u = -A(x)^{-1} \left(\left[\begin{matrix}L_f^{r_1}h_1(x)\\ \vdots \\ L_f^{r_p}h_p(x) \end{matrix}\right]+\left[\begin{matrix}\left(\frac{\partial V_1(\xi_1)}{\partial \xi_1}\right)^T \\ \left(\frac{\partial V_2(\xi_2)}{\partial \xi_2}\right)^T+\lambda \xi_3\end{matrix}
\right]\right),	
\end{align}
where $A(x)$ is defined in (\ref{eq:A(x)}), $V_1(\xi_1)$ and $V_2(\xi_2)$ can be any positive definite functions, and $\lambda>0$ is a scalar.	
\end{theorem}
\begin{proof}
Note that the functions $V_1(\xi_1)$ and $V_2(\xi_2)$ are now positive definite. Under the state feedback (\ref{eq:stabilizing state feedback}), the system becomes the system (\ref{eq:system rendered NI}) but with $v=0$. We define the storage function $\tilde V(\xi)$ of this system to be the same as that in (\ref{eq:rendered NI storage function}
) but with $V_1(\xi_1)$ and $V_2(\xi_2)$ positive definite. Therefore, $\tilde V(\xi)$ is positive definite. The inequality (\ref{eq:nonlinear OSNI ineq of the rendered system}) implies that
\begin{equation*}
\dot {\tilde V}(\xi)\leq -\epsilon\|\dot y\|^2,
\end{equation*}
where $\epsilon>0$. This means that $\dot {\tilde V}(\xi)\equiv 0$ is only possible if $\dot y \equiv 0$. This implies that $\dot \xi_1 \equiv 0$, $\dot \xi_2 \equiv 0$ and therefore $\xi_3 \equiv 0$. According to (\ref{eq:system rendered NI b}) and (\ref{eq:system rendered NI d}) and considering that $v=0$, this is only possible if $\frac{\partial V_1(\xi_1)}{\partial \xi_1}=0$ and $\frac{\partial V_2(\xi_2)}{\partial \xi_2}=0$. Since $V_1(\xi_1)$ and $V_2(\xi_2)$ can be any positive definite functions, then we can choose $V_1(\xi_1)$ and $V_2(\xi_2)$ such that $\frac{\partial V_1(\xi_1)}{\partial \xi_1}=0$ only at $\xi_1=0$ and $\frac{\partial V_2(\xi_2)}{\partial \xi_2}=0$ only at $\xi_2=0$. For example, a suitable choice is $V_1(\xi_1)=\frac{1}{2}\xi_1^T\xi_1$ and $V_2(\xi_2)=\frac{1}{2}\xi_2^T\xi_2$. Therefore, with these $V_1(\xi_1)$ and $V_2(\xi_2)$ used in the state feedback, $\dot y\equiv 0$ only at $\xi=0$. Otherwise, $\dot {\tilde V}(\xi)<0$ and the state $\xi$ will also converge to zero. Since the system is strictly minimum phase, then the internal dynamics (\ref{eq:system rendered NI a}) are LISS. According to Lemma \ref{lemma:LISS x convergence}, $z$ will converge to $0$. Therefore, the system is asymptotically stabilized locally around $x=0$.
\end{proof}
\begin{remark}
In Theorem \ref{theorem:locally stabilization with 0 new input}, if the strictly minimum phase requirement in Condition (iii) is replaced by the standard global minimum phase requirement as defined in Definition \ref{def:minimum phase}, then Theorem \ref{theorem:locally stabilization with 0 new input} still holds. This is because 0-GAS implies LISS as is proved in \cite{sontag1996new}, and hence global minimum phase implies strictly minimum phase. Reference \cite{sontag1996new} includes many equivalences and implications between different stability properties. Here, we only use one of these equivalent properties as a requirement in our stability result.
\end{remark}

\begin{theorem}\label{theorem:globally stabilization with 0 new input}
After possible output transformation (\ref{eq:output transformation}), suppose the system (\ref{eq:general system}) satisfies H1, H2 and H3. Also, suppose the system (\ref{eq:general system}) is globally strictly minimum phase. Then the system (\ref{eq:general system}) can be globally asymptotically stabilized using the state feedback (\ref{eq:stabilizing state feedback}).	
\end{theorem}
\begin{proof}
See the proof of Theorem \ref{theorem:locally stabilization with 0 new input}, using Lemmas \ref{lemma:global normal form} and \ref{lemma:ISS x convergence} instead of Lemmas \ref{lemma:normal form conditions} and \ref{lemma:LISS x convergence} for a global result.	
\end{proof}

\section{CONTROLLER SYNTHESIS FOR A SYSTEM WITH NONLINEAR NI UNCERTAINTY}\label{sec:NI uncertainty}
Suppose a system of the form (\ref{eq:general system}) has uncertainty that can be modelled as a nonlinear NI system. Denote the uncertainty as $H_c$. The system model of $H_c$ is
\begin{subequations}\label{eq:uncertainty}
\begin{align}
H_c:\quad\dot x_c =& f_c(x_c,u_c),\\
y_c =& h_c(x_c),
\end{align}	
\end{subequations}
where $x_c\in \mathbb R^{n_c}$ is the state, $u_c\in \mathbb R^p$ is the input, and $y_c\in \mathbb R^p$ is the output, $f_c:\mathbb R^{n_c}\times \mathbb R^p \to \mathbb R^{n_c}$ is a Lipschitz continuous function and $h_c:\mathbb R^{n_c} \to \mathbb R^p$ is a class $C^1$ function. Suppose the system has at least one equilibrium. Then without loss of generality, we can assume $f_c(0,0)=0$ and $h_c(0)=0$ after a possible coordinate shift.

When full state information is available, we aim to stabilize the uncertain system using a state feedback controller as shown in the left-hand side (LHS) of Fig.~\ref{fig:controller synthesis}.

\begin{figure}[h!]
\centering
\psfrag{delta}{\hspace{0.1cm}$H_c$}
\psfrag{nominal}{\hspace{-0.05cm}\small Nominal}
\psfrag{plant}{\hspace{-0.2cm}\small Plant $\Sigma$}
\psfrag{controller}{\hspace{-0.1cm}\small Controller}
\psfrag{closed-loop}{\hspace{-0.1cm}\small Closed-Loop}
\psfrag{w}{$w$}
\psfrag{x}{$x$}
\psfrag{y}{$y$}
\psfrag{u}{$u$}
\psfrag{R_s}{\hspace{0.1cm}\small $H_p$}
\psfrag{+}{\scriptsize$+$}
\includegraphics[width=8.5cm]{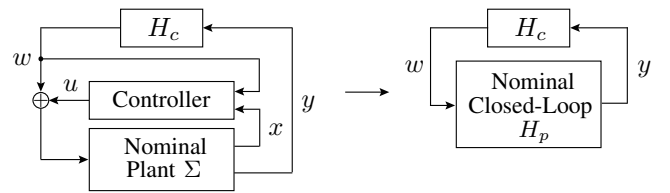}
\caption{A feedback control system. The nominal plant $\Sigma$ has a plant uncertainty $H_c$, which can be described as a nonlinear NI system. Under suitable assumptions, we can find a state feedback control input such that the resulting closed-loop system is guaranteed to be asymptotically stable.}
\label{fig:controller synthesis}
\end{figure}

The interconnection can be described by the following equations:
\begin{subequations}\label{eq:closed-loop}
\begin{align}
\dot x =& f(x)+g(x)(u+w),\label{eq:closed-loop a}\\
y =& h(x),\label{eq:closed-loop b}\\
\dot x_c =& f_c(x_c,u_c),\\
y_c =& h_c(x_c),\label{eq:closed-loop d}\\
w=&y_c,\label{eq:closed-loop e}\\
u_c=&y.\label{eq:closed-loop f}
\end{align}
\end{subequations}

\begin{theorem}\label{theorem:stabilization}
	Suppose the nominal plant $\Sigma$ of the form (\ref{eq:general system}) is strictly minimal phase and has relative degree less than or equal to two around $x=0$ and the distribution
\begin{equation*}\label{eq:G}
G = span\{g^1,g^2,\cdots,g^p\}	
\end{equation*}
is involutive. Let $\xi_1 = [y_1^T,\cdots y_{p_1}^T]^T$ and $\xi_2 = [y_{p_1+1}^T,\cdots, y_p^T]^T$ denote the vectors containing the output entries corresponding to the ones and twos in the vector relative degree, respectively. Let $\xi_3 = \dot \xi_2$. Suppose that the systems (\ref{eq:uncertainty}) is nonlinear NI with storage function $V_c(x_c)$. If there exist positive definite functions $V_1(\xi_1)$ and $V_2(\xi_2)$ such that the function defined as
\begin{equation}\label{eq:storage function closed-loop}
W(\xi,x_c)=V_1(\xi_1)+V_2(\xi_2)+\frac{1}{2}\xi_3^T\xi_3+V_c(x_c)-h_c(x_c)^T\left[\begin{matrix}\xi_1\\ \xi_2\end{matrix}
\right]
\end{equation}
is positive definite, then the system (\ref{eq:closed-loop}) is locally asymptotically stabilized by the state feedback control law
\begin{align}\label{eq:state feedback control}
u = A(x)^{-1}&\left( \left(I-\left[\begin{matrix}L_gL_f^{r_1-1}h_1(x)\\ \vdots \\ L_gL_f^{r_p-1}h_p(x) \end{matrix}\right]\right)w\right.\notag\\
&-\left.\left[\begin{matrix}L_f^{r_1}h_1(x)\\ \vdots \\ L_f^{r_p}h_p(x) \end{matrix}\right]
-\left[\begin{matrix}\left(\frac{\partial V_1(\xi_1)}{\partial \xi_1}\right)^T \\ \left(\frac{\partial V_2(\xi_2)}{\partial \xi_2}\right)^T+\lambda \xi_3\end{matrix}
\right]\right),	
\end{align}
where $A(x)$ is defined in (\ref{eq:A(x)}), $w\in \mathbb R^p$ is the output of the uncertainty $H_c$, and $\lambda$ is a positive scalar.
\end{theorem}
\begin{proof}
With the state feedback control (\ref{eq:state feedback control}), the nominal plant (\ref{eq:closed-loop a}), (\ref{eq:closed-loop b}) now becomes the nominal closed-loop system $H_p$, as shown on the right-hand side of Fig.~\ref{fig:controller synthesis}. The system $H_p$ has a normal form similar to (\ref{eq:system rendered NI}), where $v_1$ and $v_2$ are replaced by $w_1$ and $w_2$, respectively. $w=[w_1^T\ w_2^T]^T$.  According to Theorem \ref{theorem:render NI}, the system $H_p$ is a nonlinear OSNI system with the storage function
\begin{equation*}
V(z,\xi)=\tilde V(\xi) = V_1(\xi_1)+V_2(\xi_2)+\frac{1}{2}\xi_3^T\xi_3,
\end{equation*}
which is positive semidefinite because $V(z,\xi)=W(\xi,0)$.
This storage function satisfies the nonlinear OSNI inequality:
\begin{equation*}
\dot V(z,\xi)\leq w^T \dot y-\epsilon \| \dot y \|^2,
\end{equation*}
where $\epsilon>0$ quantifies the output strictness of the system. For the interconnection of the nonlinear OSNI system $H_p$ and the nonlinear NI system $H_c$ that is shown on the RHS of Fig.~\ref{fig:controller synthesis}, we use the function (\ref{eq:storage function closed-loop}) as a Lyapunov storage function. Since $H_c$ is a nonlinear NI uncertainty with storage function $V_c(x_c)$, we have that
\begin{equation*}
\dot V_c(x_c)\leq u_c^T \dot y_c	
\end{equation*}
according to Definition \ref{def:nonlinear NI}. We show in the following that the stability of this interconnection is guaranteed according to Lyapunov's stability theorem. We have that
\begin{align}
\dot W &(\xi,x_c)\notag\\
 =& \dot V(z,\xi)+\dot V_c(x_c)-\dot h_c(x_c)^T\left[\begin{matrix}\xi_1\\ \xi_2\end{matrix}
\right] - h_c(x_c)^T\left[\begin{matrix}\dot \xi_1\\ \dot \xi_2\end{matrix}
\right]\notag\\
\leq & w^T\dot y - \epsilon \|\dot y\|^2 +u_c^T\dot h_c(x_c)-\dot h_c(x_c)^T\left[\begin{matrix}\xi_1\\ \xi_2\end{matrix}
\right]\notag\\
& - h_c(x_c)^T\left[\begin{matrix}\dot \xi_1\\ \dot \xi_2\end{matrix}
\right]\notag\\
=& - \epsilon \|\dot y\|^2\notag\\
\leq & 0,\notag
\end{align}
where the equality also uses (\ref{eq:closed-loop b}), (\ref{eq:closed-loop d}), (\ref{eq:closed-loop e}) and (\ref{eq:closed-loop f}).
Therefore, $\dot W(y,x_2)\equiv 0$ is only possible if $\dot y\equiv 0$. In this case, $\dot \xi_1\equiv 0$ and $\dot \xi_2 \equiv 0$. This implies that $\xi_1,\xi_2$ remain constant and $\xi_3\equiv 0$. Also, according to (\ref{eq:system rendered NI b}) and (\ref{eq:system rendered NI d}), we have that $w_1 = \left(\frac{\partial V_1(\xi_1)}{\partial \xi_1}\right)^T$ and $w_2 = \left(\frac{\partial V_2(\xi_2)}{\partial \xi_2}\right)^T$, which now both remain constant. Considering the setting (\ref{eq:closed-loop e}) of the interconnection, the output $y_c$ of the nonlinear NI uncertainty $H_c$ now remains constant. For the nonlinear NI uncertainty $H_c$, its input $u_c=y$ and output $y_c=w$ both remain constant. Moreover, given constant input $[\xi_1^T\ \xi_2^T]^T$ to the system $H_c$, we get constant output $[\frac{\partial V_1(\xi_1)}{\partial \xi_1}\ \frac{\partial V_2(\xi_2)}{\partial \xi_2}]^T$. We prove in the following that this situation can be avoided. Suppose steady-state input-output relationship of the uncertainty system $H_c$ is described by some function $\bar y = \kappa (\bar u)$, where $\bar y$ and $\bar u$ are the constant output and input respectively in steady state. Then we can always add additional positive definite functions to $V_1(\xi_1)$ and $V_2(\xi_2)$ such that the curve of $[\frac{\partial V_1(\xi_1)}{\partial \xi_1}\ \frac{\partial V_2(\xi_2)}{\partial \xi_2}]^T$ intersects with the curve of $\kappa(\left[\begin{matrix}	\xi_1 \\ \xi_2 \end{matrix}
\right])$ only at the origin. That is, the entire closed-loop system cannot remain in a steady state unless $y \equiv 0$. Note that this will not affect the positive definiteness of $W(\xi,x_c)$ because we are adding positive definite functions to $V_1(\xi_1)$ and $V_2(\xi_2)$. Otherwise if $y$ does not identically remain zero, $\dot W(\xi,x_c)$ cannot remain zero. It will keep decreasing until $y=0$ and $x_c=0$. This means that eventually $\xi\to 0$. Since $\Sigma$ is strictly minimal phase, $z$ will also converge to zero. Therefore, the closed-loop system as described in (\ref{eq:closed-loop}) is asymptotically stabilized locally around $x=0$.
\end{proof}

\begin{theorem}\label{theorem:global stabilization}
	Suppose the nominal plant $\Sigma$ of the form (\ref{eq:general system}) is globally strictly minimal phase and satisfies H1, H2 and H3. Let $\xi_1 = [y_1^T,\cdots y_{p_1}^T]^T$ and $\xi_2 = [y_{p_1+1}^T,\cdots, y_p^T]^T$ denote the vectors containing the output entries corresponding to the ones and twos in the vector relative degree, respectively. Let $\xi_3 = \dot \xi_2$. Suppose that the systems (\ref{eq:uncertainty}) is nonlinear NI with storage function $V_c(x_c)$. If there exist positive definite functions $V_1(\xi_1)$ and $V_2(\xi_2)$ such that the function (\ref{eq:storage function closed-loop}) is positive definite, then the system (\ref{eq:closed-loop}) is globally robustly stabilized by the state feedback control law (\ref{eq:state feedback control}).
\end{theorem}
\begin{proof}
See the proof of Theorem \ref{theorem:stabilization}, using Lemmas \ref{lemma:global normal form} and \ref{lemma:ISS x convergence} instead of Lemmas \ref{lemma:normal form conditions} and \ref{lemma:LISS x convergence}.	
\end{proof}

\section{EXAMPLE}\label{sec:example}
In this section, we illustrate the stabilization process for a system with nonlinear NI uncertainty. We show that if the conditions in Theorems \ref{theorem:stabilization} and \ref{theorem:global stabilization} are satisfied by choosing suitable state feedback, the uncertain system can be asymptotically stabilized.

As the conditions for the existence of normal forms provided in Lemma \ref{lemma:normal form conditions} and \ref{lemma:global normal form} are not the main focus of this work, we investigate an uncertain system whose nominal plant is already in its normal form. Consider an uncertain system as shown in the LHS of Fig.~\ref{fig:controller synthesis}. Suppose the nominal plant $\Sigma$ has the state-space model:
\begin{subequations}\label{eq:example nominal plant}
\begin{align}
\Sigma:\quad \dot z =& -z-z^3+\xi_1^2,\label{eq:example nominal plant a}\\
\dot \xi_1 =& \sin z + (u_1+w_1),\\
\dot \xi_2 = & \xi_3,\\
\dot \xi_3 =& \xi_1+\xi_2^2+\xi_3+(u_2+w_2),\\
y = & \left[\begin{matrix}\xi_1\\ \xi_2	
\end{matrix}
\right],	
\end{align}
\end{subequations}
where $x = [z\ \xi_1 \ \xi_2\ \xi_3]^T$, $u=[u_1\ u_2]^T$ and $y = [\xi_1\ \xi_2]^T$ are the state, nominal input and output of the system, respectively. $w=[w_1\  w_2]^T$ is the output of the plant uncertainty. The sum of the nominal input and the uncertainty output; i.e., $u+w$, acts as the actual input of the system (\ref{eq:example nominal plant}). Here $z,\xi_1,\xi_2,\xi_3,u_1,u_2,w_1,w_2\in \mathbb R$. This system is globally strictly minimum phase because the internal dynamics (\ref{eq:example nominal plant a}) are ISS with respect to the state $\xi$ (see for example \cite[Theorem 4.19]{khalil2002nonlinear}, using the Lyapunov function $V(z)=z^2$).

 Suppose the plant uncertainty is a nonlinear NI system and has the model:
\begin{subequations}\label{eq:example NI uncertainty}
\begin{align}
\dot x_{c1} =& -x_{c1}+u_{c1},\\
\dot x_{c2} =& -x_{c2}^3+u_{c2},\\
y_{c} =& \left[\begin{matrix}x_{c1}\\x_{c2}
\end{matrix}
\right],
\end{align}
\end{subequations}
where $x_c = [x_{c1}\ x_{c2}]^T$, $u_c = [u_{c1}\ u_{c2}]^T$ and $y_c = [x_{c1}\ x_{c2}]^T$ are the state, input and output of the system, respectively. Here, $x_{c1}, x_{c2}, u_{c1}, u_{c2}\in \mathbb R$.
The system (\ref{eq:example NI uncertainty}) is nonlinear NI with the positive definite storage function
\begin{equation*}
V_c(x_c) = \frac{1}{2}x_{c1}^2+\frac{1}{4}x_{c2}^4,
\end{equation*}
which satisfies the nonlinear NI property
\begin{equation*}
\dot V(x_c) \leq u_c^T \dot y_c.	
\end{equation*}
The interconnection between the nominal plant (\ref{eq:example nominal plant}) and the plant uncertainty, as shown in Fig.~\ref{fig:controller synthesis}, is
\begin{equation}\label{eq:example interconnection}
	u_c = y; \quad \textnormal{and} \quad w = y_c.
\end{equation}
We choose positive definite functions $V_1(\xi_1)$ and $V_2(\xi_2)$ to be
\begin{equation*}
V_1(\xi_1)= \xi_1^2\quad \textnormal{and} \quad V_2(\xi_2) = \xi_2^{\frac{4}{3}},	
\end{equation*}
which makes the storage function of the entire system, constructed using the formula (\ref{eq:storage function closed-loop}), positive definite. The storage function is
\begin{equation}\label{eq:example storage function W}
W(\xi,x_c) = \xi_1^2+\xi_2^{\frac{4}{3}} + \frac{1}{2}\xi_3^2+\frac{1}{2}x_{c1}^2+\frac{1}{4}x_{c2}^4-\xi_1x_{c1}-\xi_2x_{c2}.
\end{equation}
The corresponding state feedback control input, according to (\ref{eq:state feedback control}), is
\begin{equation}\label{eq:example state feedback}
u = -\left[\begin{matrix}\sin{z}+2\xi_1 \\ \xi_1+\xi_2^2+\frac{4}{3}\xi_2^{\frac{1}{3}}+2\xi_3 \end{matrix}\right].	
\end{equation}
We show in the following that this state feedback control law stabilizes the system. Under the state feedback (\ref{eq:example state feedback}), the nominal plant (\ref{eq:example nominal plant}) now becomes the nominal closed-loop system, as shown on the right-hand side of Fig.~\ref{fig:controller synthesis}. It has the following system model
\begin{subequations}
\begin{align}
\Sigma:\quad \dot z =& -z-z^3+\xi_1^2,\label{eq:example nominal plant new a}\\
\dot \xi_1 =& -2\xi_1 + w_1,\label{eq:example nominal plant new b}\\
\dot \xi_2 = & \xi_3,\label{eq:example nominal plant new c}\\
\dot \xi_3 =& -\frac{4}{3}\xi_2^{\frac{1}{3}}-\xi_3+w_2,\label{eq:example nominal plant new d}\\
y = & \left[\begin{matrix}\xi_1\\ \xi_2	
\end{matrix}
\right].\label{eq:example nominal plant new e}	
\end{align}
\end{subequations}
Using Lyapunov's direct method, the time derivative of the storage function (\ref{eq:example storage function W}) is
\begin{align}
	\dot W & (\xi,x_c)\notag\\
	=& -5\xi_1^2+6\xi_1x_{c1}-2x_{c1}^2-\xi_3^2-x_{c2}^6+2x_{c2}^3\xi_2-\xi_2^2\notag\\
	=& -\|\dot y\|^2-(\xi_1-x_{c1})^2-(\xi_2-x_{c2}^3)^2\notag\\
	\leq & 0.\label{eq:dot W example}
\end{align}
It can observed from (\ref{eq:dot W example}) that $\dot W(\xi,x_c) \equiv 0$ only if $\dot y=0$, $\xi_1=x_{c1}$ and $\xi_2 = x_{c2}^3$. This implies that $2\xi_1 = w_1=x_{c1}=\xi_1=0$ and $\frac{4}{3}\xi_2^{\frac{1}{3}}=w_2=x_{c2}=\xi_2^{\frac{1}{3}}=0$, where (\ref{eq:example interconnection}) and (\ref{eq:example nominal plant new b})-(\ref{eq:example nominal plant new e}) are also used. This implies that $W(\xi,x_c)$ will keep decreasing until $\xi=0$ and $x_c=0$. According to Lemma \ref{lemma:ISS x convergence}, $z$ will also converge to zero. We also simulate this uncertain system under the state feedback (\ref{eq:example state feedback}). Let the initial state of the nominal plant (\ref{eq:example nominal plant}) be $x(0) = [z(0)\ \xi_1(0) \ \xi_2(0) \ \xi_3(0)]^T = [10 \ 3 \ -5 \ 7]$ and the initial state of the plant uncertainty (\ref{eq:example NI uncertainty}) be $x_c(0) = [x_{c1}(0)\ x_{c2}(0)]^T = [-8 \ 2]$. Fig.~\ref{fig:simulation} shows the state trajectories of the nominal closed-loop system $H_p$; i.e., the nominal plant (\ref{eq:example nominal plant}) under the state feedback (\ref{eq:example state feedback}). Despite the presence of the nonlinear NI uncertainty $H_c$ as described by (\ref{eq:example NI uncertainty}), the plant states still converge to zero.

\begin{figure}[h!]
\centering
\psfrag{State}{State}
\psfrag{time (s)}{Time (s)}
\psfrag{States of the Nominal Plant of the Uncertain System}{\hspace{0.3cm}State Trajectories of the Nominal Plant}
\psfrag{z}{\scriptsize$z$}
\psfrag{x1}{\scriptsize$\xi_1$}
\psfrag{x2}{\scriptsize$\xi_2$}
\psfrag{x3}{\scriptsize$\xi_3$}
\psfrag{0}{\small$0$}
\psfrag{2}{\small$2$}
\psfrag{4}{\small$4$}
\psfrag{6}{\small$6$}
\psfrag{8}{\small$8$}
\psfrag{10}{\small$10$}
\psfrag{15}{\small$15$}
\psfrag{20}{\small$20$}
\psfrag{-5}{\hspace{-0.162cm}\small$-5$}
\psfrag{5}{\small$5$}
\psfrag{10}{\small$10$}
\includegraphics[width=9cm]{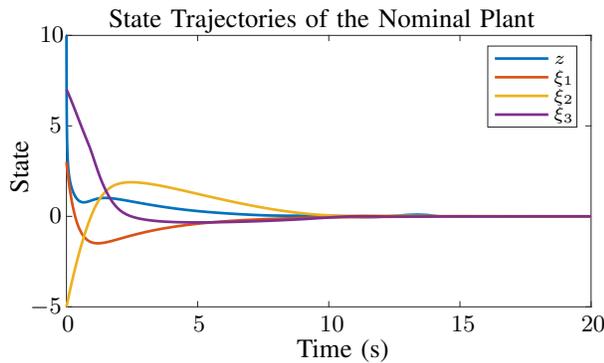}
\caption{State trajectories of the uncertain system (\ref{eq:example nominal plant}) under the state feedback control (\ref{eq:example state feedback}) constructed according to Theorem \ref{theorem:global stabilization}. Starting from nonzero initial values, the states of the nominal closed-loop system converge to zero, despite the presence of a nonlinear NI plant uncertainty (\ref{eq:example NI uncertainty}).}
\label{fig:simulation}
\end{figure}

\section{CONCLUSION}\label{sec:conclusion}
This paper investigates a state feedback stabilization problem using nonlinear NI systems theory for affine nonlinear systems with relative degree less than or equal to two. For such a system that also has a normal form, we provide a state feedback control law that makes it nonlinear NI or OSNI. In this case, if the internal dynamics of the system are ISS, then there exists state feedback control that stabilizes the system. In the case that the system has a nonlinear NI uncertainty, there exists a state feedback control law such that stabilizes the system if a positive definiteness-like assumption is satisfied for the storage function of the closed-loop system. A numerical example is provided to illustrate the process of stabilizing a system with a nonlinear NI uncertainty. Simulation shows that stabilization is achieved, as expected according to the proposed results.

\bibliographystyle{IEEEtran}

\begin{thebibliography}{10}
\providecommand{\url}[1]{#1}
\csname url@samestyle\endcsname
\providecommand{\newblock}{\relax}
\providecommand{\bibinfo}[2]{#2}
\providecommand{\BIBentrySTDinterwordspacing}{\spaceskip=0pt\relax}
\providecommand{\BIBentryALTinterwordstretchfactor}{4}
\providecommand{\BIBentryALTinterwordspacing}{\spaceskip=\fontdimen2\font plus
\BIBentryALTinterwordstretchfactor\fontdimen3\font minus
  \fontdimen4\font\relax}
\providecommand{\BIBforeignlanguage}[2]{{%
\expandafter\ifx\csname l@#1\endcsname\relax
\typeout{** WARNING: IEEEtran.bst: No hyphenation pattern has been}%
\typeout{** loaded for the language `#1'. Using the pattern for}%
\typeout{** the default language instead.}%
\else
\language=\csname l@#1\endcsname
\fi
#2}}
\providecommand{\BIBdecl}{\relax}
\BIBdecl

\bibitem{lanzon2008stability}
A.~Lanzon and I.~R. Petersen, ``Stability robustness of a feedback
  interconnection of systems with negative imaginary frequency response,''
  \emph{IEEE Transactions on Automatic Control}, vol.~53, no.~4, pp.
  1042--1046, 2008.

\bibitem{petersen2010feedback}
I.~R. Petersen and A.~Lanzon, ``Feedback control of negative-imaginary
  systems,'' \emph{IEEE Control Systems Magazine}, vol.~30, no.~5, pp. 54--72,
  2010.

\bibitem{preumont2018vibration}
A.~Preumont, \emph{Vibration control of active structures: an
  introduction}.\hskip 1em plus 0.5em minus 0.4em\relax Springer, 2018, vol.
  246.

\bibitem{halim2001spatial}
D.~Halim and S.~R. Moheimani, ``Spatial resonant control of flexible
  structures-application to a piezoelectric laminate beam,'' \emph{IEEE
  Transactions on Control Systems Technology}, vol.~9, no.~1, pp. 37--53, 2001.

\bibitem{pota2002resonant}
H.~Pota, S.~R. Moheimani, and M.~Smith, ``Resonant controllers for smart
  structures,'' \emph{Smart Materials and Structures}, vol.~11, no.~1, p.~1,
  2002.

\bibitem{brogliato2007dissipative}
B.~Brogliato, R.~Lozano, B.~Maschke, and O.~Egeland, ``Dissipative systems
  analysis and control,'' \emph{Theory and Applications}, vol.~2, 2007.

\bibitem{xiong2010negative}
J.~Xiong, I.~R. Petersen, and A.~Lanzon, ``A negative imaginary lemma and the
  stability of interconnections of linear negative imaginary systems,''
  \emph{IEEE Transactions on Automatic Control}, vol.~55, no.~10, pp.
  2342--2347, 2010.

\bibitem{song2012negative}
Z.~Song, A.~Lanzon, S.~Patra, and I.~R. Petersen, ``A negative-imaginary lemma
  without minimality assumptions and robust state-feedback synthesis for
  uncertain negative-imaginary systems,'' \emph{Systems \& Control Letters},
  vol.~61, no.~12, pp. 1269--1276, 2012.

\bibitem{mabrok2014generalizing}
M.~A. Mabrok, A.~G. Kallapur, I.~R. Petersen, and A.~Lanzon, ``Generalizing
  negative imaginary systems theory to include free body dynamics: Control of
  highly resonant structures with free body motion,'' \emph{IEEE Transactions
  on Automatic Control}, vol.~59, no.~10, pp. 2692--2707, 2014.

\bibitem{wang2015robust}
J.~Wang, A.~Lanzon, and I.~R. Petersen, ``Robust cooperative control of
  multiple heterogeneous negative-imaginary systems,'' \emph{Automatica},
  vol.~61, pp. 64--72, 2015.

\bibitem{bhowmick2017lti}
P.~Bhowmick and S.~Patra, ``On {LTI} output strictly negative-imaginary
  systems,'' \emph{Systems \& Control Letters}, vol. 100, pp. 32--42, 2017.

\bibitem{mabrok2013spectral}
M.~A. Mabrok, A.~G. Kallapur, I.~R. Petersen, and A.~Lanzon, ``Spectral
  conditions for negative imaginary systems with applications to
  nanopositioning,'' \emph{IEEE/ASME Transactions on Mechatronics}, vol.~19,
  no.~3, pp. 895--903, 2013.

\bibitem{das2014mimo}
S.~K. Das, H.~R. Pota, and I.~R. Petersen, ``A {MIMO} double resonant
  controller design for nanopositioners,'' \emph{IEEE Transactions on
  Nanotechnology}, vol.~14, no.~2, pp. 224--237, 2014.

\bibitem{das2014resonant}
------, ``Resonant controller design for a piezoelectric tube scanner: A mixed
  negative-imaginary and small-gain approach,'' \emph{IEEE Transactions on
  Control Systems Technology}, vol.~22, no.~5, pp. 1899--1906, 2014.

\bibitem{das2015multivariable}
------, ``Multivariable negative-imaginary controller design for damping and
  cross coupling reduction of nanopositioners: a reference model matching
  approach,'' \emph{IEEE/ASME Transactions on Mechatronics}, vol.~20, no.~6,
  pp. 3123--3134, 2015.

\bibitem{cai2010stability}
C.~Cai and G.~Hagen, ``Stability analysis for a string of coupled stable
  subsystems with negative imaginary frequency response,'' \emph{IEEE
  Transactions on Automatic Control}, vol.~55, no.~8, pp. 1958--1963, 2010.

\bibitem{rahman2015design}
M.~A. Rahman, A.~Al~Mamun, K.~Yao, and S.~K. Das, ``Design and implementation
  of feedback resonance compensator in hard disk drive servo system: A mixed
  passivity, negative-imaginary and small-gain approach in discrete time,''
  \emph{Journal of Control, Automation and Electrical Systems}, vol.~26, no.~4,
  pp. 390--402, 2015.

\bibitem{bhikkaji2011negative}
B.~Bhikkaji, S.~R. Moheimani, and I.~R. Petersen, ``A negative imaginary
  approach to modeling and control of a collocated structure,'' \emph{IEEE/ASME
  Transactions on Mechatronics}, vol.~17, no.~4, pp. 717--727, 2011.

\bibitem{ghallab2018extending}
A.~G. Ghallab, M.~A. Mabrok, and I.~R. Petersen, ``Extending negative imaginary
  systems theory to nonlinear systems,'' in \emph{2018 IEEE Conference on
  Decision and Control (CDC)}.\hskip 1em plus 0.5em minus 0.4em\relax IEEE,
  2018, pp. 2348--2353.

\bibitem{shi2021robust}
K.~Shi, I.~G. Vladimirov, and I.~R. Petersen, ``Robust output feedback
  consensus for networked identical nonlinear negative-imaginary systems,''
  \emph{IFAC-PapersOnLine}, vol.~54, no.~9, pp. 239--244, 2021.

\bibitem{ghallab2022negative}
A.~G. Ghallab and I.~R. Petersen, ``Negative imaginary systems theory for
  nonlinear systems: A dissipativity approach,'' \emph{arXiv preprint
  arXiv:2201.00144}, 2022.

\bibitem{angeli2006systems}
D.~Angeli, ``Systems with counterclockwise input-output dynamics,'' \emph{IEEE
  Transactions on automatic control}, vol.~51, no.~7, pp. 1130--1143, 2006.

\bibitem{shi2020robustc}
K.~Shi, I.~R. Petersen, and I.~G. Vladimirov, ``Output feedback consensus for
  networked heterogeneous nonlinear negative-imaginary systems with free body
  motion,'' \emph{arXiv preprint arXiv:2011.14610v2}, 2021.

\bibitem{bhowmickoutput}
P.~Bhowmick and A.~Lanzon, ``Output strictly negative imaginary systems and its
  connections to dissipativity theory,'' in \emph{2019 IEEE 58th Conference on
  Decision and Control (CDC)}.\hskip 1em plus 0.5em minus 0.4em\relax IEEE,
  2019, pp. 6754--6759.

\bibitem{kokotovic1989positive}
P.~Kokotovic and H.~Sussmann, ``A positive real condition for global
  stabilization of nonlinear systems,'' \emph{Systems \& Control Letters},
  vol.~13, no.~2, pp. 125--133, 1989.

\bibitem{saberi1990global}
A.~Saberi, P.~Kokotovic, and H.~Sussmann, ``Global stabilization of partially
  linear composite systems,'' \emph{SIAM Journal on Control and Optimization},
  vol.~28, no.~6, pp. 1491--1503, 1990.

\bibitem{byrnes1991passivity}
C.~Byrnes, A.~Isidori, and J.~Willems, ``Passivity, feedback equivalence, and
  the global stabilization of minimum phase nonlinear systems,'' \emph{IEEE
  Transactions on Automatic Control}, vol.~36, no.~11, pp. 1228--1240, 1991.

\bibitem{byrnes1991asymptotic}
C.~I. Byrnes and A.~Isidori, ``Asymptotic stabilization of minimum phase
  nonlinear systems,'' \emph{IEEE Transactions on Automatic Control}, vol.~36,
  no.~10, pp. 1122--1137, 1991.

\bibitem{santosuosso1997passivity}
G.~Santosuosso, ``Passivity of nonlinear systems with input-output
  feedthrough,'' \emph{Automatica}, vol.~33, no.~4, pp. 693--697, 1997.

\bibitem{lin1995feedback}
W.~Lin, ``Feedback stabilization of general nonlinear control systems: a
  passive system approach,'' \emph{Systems \& Control Letters}, vol.~25, no.~1,
  pp. 41--52, 1995.

\bibitem{jiang1996passification}
Z.-P. Jiang, D.~J. Hill, and A.~L. Fradkov, ``A passification approach to
  adaptive nonlinear stabilization,'' \emph{Systems \& Control Letters},
  vol.~28, no.~2, pp. 73--84, 1996.

\bibitem{shi2021negative}
K.~Shi, I.~R. Petersen, and I.~G. Vladimirov, ``Negative imaginary state
  feedback equivalence for systems of relative degree one and relative degree
  two,'' in \emph{2021 60th IEEE Conference on Decision and Control (CDC)},
  2021, pp. 3948--3953.

\bibitem{shi2021necessary}
------, ``Necessary and sufficient conditions for state feedback equivalence to
  negative imaginary systems,'' \emph{arXiv preprint arXiv:2109.11273}, 2021.

\bibitem{isidori1995nonlinear}
A.~Isidori, E.~Sontag, and M.~Thoma, \emph{Nonlinear control systems}.\hskip
  1em plus 0.5em minus 0.4em\relax Springer, 1995, vol.~3.

\bibitem{khalil2002nonlinear}
H.~K. Khalil, \emph{Nonlinear systems}.\hskip 1em plus 0.5em minus 0.4em\relax
  Prentice Hall Upper Saddle River, NJ, 2002, vol.~3.

\bibitem{sontag1989smooth}
E.~D. Sontag \emph{et~al.}, ``Smooth stabilization implies coprime
  factorization,'' \emph{IEEE transactions on automatic control}, vol.~34,
  no.~4, pp. 435--443, 1989.

\bibitem{sontag2008input}
E.~D. Sontag, ``Input to state stability: Basic concepts and results,'' in
  \emph{Nonlinear and optimal control theory}.\hskip 1em plus 0.5em minus
  0.4em\relax Springer, 2008, pp. 163--220.

\bibitem{liberzon2002output}
D.~Liberzon, A.~S. Morse, and E.~D. Sontag, ``Output-input stability and
  minimum-phase nonlinear systems,'' \emph{IEEE Transactions on Automatic
  Control}, vol.~47, no.~3, pp. 422--436, 2002.

\bibitem{praly1993stabilization}
L.~Praly and Z.-P. Jiang, ``Stabilization by output feedback for systems with
  {ISS} inverse dynamics,'' \emph{Systems \& Control Letters}, vol.~21, no.~1,
  pp. 19--33, 1993.

\bibitem{liberzon2003switching}
D.~Liberzon, \emph{Switching in systems and control}.\hskip 1em plus 0.5em
  minus 0.4em\relax Springer Science \& Business Media, 2003.

\bibitem{sontag1996new}
E.~D. Sontag and Y.~Wang, ``New characterizations of input-to-state
  stability,'' \emph{IEEE Transactions on Automatic Control}, vol.~41, no.~9,
  pp. 1283--1294, 1996.

\end{thebibliography}

\end{document}